\documentclass[lettersize,journal]{IEEEtran}

\usepackage{graphicx} 

\usepackage{tikz}
\usetikzlibrary{positioning, arrows.meta}

\usepackage{graphicx}%
\usepackage{multirow}%
\usepackage{amsmath,amssymb,amsfonts}%
\usepackage{amsthm}%
\usepackage{mathrsfs}%
\usepackage{xcolor}%
\usepackage{textcomp}%
\usepackage{manyfoot}%
\usepackage{booktabs}%
\usepackage{listings}%
\usepackage[english]{babel}
\usepackage[numbers]{natbib}
\usepackage{braket}

\newcommand{\speed}{\textsf{c}}

\newcommand{\maxdist}{\delta}
\newcommand{\secparam}{\lambda}

\newcommand{\timethreshold}{\Delta}

\newtheorem{theorem}{Theorem}[section]

\newtheorem{lemma}[theorem]{Lemma}

\newcommand{\gen}{\textsf{Gen}}
\newcommand{\readf}{\textsf{Read}}
\newcommand{\chk}{\textsf{CHK}}

\newcommand{\soteria}{\textsf{Soteria}}

\usepackage{algorithm}
\usepackage{algorithmic}

\usepackage[
n,
operators,
advantage,
sets,
adversary,
landau,
probability,
notions,	
logic,
ff,
mm,
primitives,
events,
complexity,
oracles,
asymptotics,
keys]{cryptocode}

\usepackage{subcaption}
\newtheorem{definition}{Definition}

\usepackage{tikz,pgfplots}
\usepackage{placeins}
\usetikzlibrary{calc}
\usetikzlibrary{intersections,patterns,angles,quotes}
\usepgfplotslibrary{fillbetween}
\usepackage{tikz-cd}
\usepackage{graphicx}
\usepackage{url}

\usepackage{hyperref}

\ifCLASSINFOpdf
\else
\fi
\begin{document}
%
\title{Approaches to Quantum Remote Memory Attestation}

\author{
    Jesse Laeuchli,
    Rolando Trujillo Rasua\\
    \thanks{Jesse Laeuchli is with the School of Computer Science and Engineering, UNSW Sydney, Sydney, Australia (e-mail: j.laeuchli@unsw.edu.au). ORCID: \href{https://orcid.org/0000-0001-9970-9105}{0000-0001-9970-9105}.}
    \thanks{Rolando Trujillo Rasua is with the Dep. Enginy. Informàtica i Matemàtiques, Rovira i Virgili University, Tarragona, Spain (e-mail: rolando.trujillo@urv.cat). ORCID: \href{https://orcid.org/0000-0002-8714-4626}{0000-0002-8714-4626}.}
}


%


\maketitle
\pagestyle{plain}

\begin{abstract}
    In this article we uncover flaws and pitfalls of a quantum-based remote memory attestation procedure for Internet-of-Things devices. We also show  limitations of quantum memory that suggests 
    the attestation problem for quantum memory is fundamentally different to the attestation problem for classical memory, even when the devices can perform quantum computation. 
    The identified problems are of interest for quantum-based security protocol designers in general, particularly those dealing with corrupt devices. 
    Finally, we make use of the lessons learned to design a quantum-based attestation system for classical memory with improved communication efficiency and security. 
\end{abstract}

\begin{IEEEkeywords}
Internet of Things (IoT) security, quantum computing (QC),
remote attestation protocol.
\end{IEEEkeywords}

%
\IEEEpeerreviewmaketitle

\section{Introduction}
\label{sec-introduction}

Determining whether a remote embedded device is in a secure state is an important problem in computer security. A solution to this problem is \emph{remote memory attestation}: a two-party communication protocol whereby a device is asked to prove that its memory is in a given state. If the device successfully passes the protocol, then the analyst can be confident that it has not been compromised. 
 
Developing secure protocols for remote attestation is inherently difficult due to the assumption that the proving device is compromised. This assumption breaks conventional security protocols, which require the communicating parties to either share a cryptographic secret, follow the protocol specification, or securely connect to a trusted-third-party.
Because the remote device is potentially infected with malware, 
the state-of-the-practice in remote memory attestation is to deploy  
secure hardware on the proving device \cite{SARA,RESERVE,ROMAN2023,Conti2019} in such a way that at least one of the three above requirements is satisfied.

Trusted hardware, however, increases the manufacturing cost, which typically cannot be afforded by resource-constrained IoT systems \cite{5504802}, and depends on engineering and supply chain guarantees rather than on information-theoretic security. If a security vulnerability is later found with the installed hardware, then all deployed devices must be recalled. Unlike software vulnerabilities, hardware vulnerabilities cannot be hot fixed.

Orthogonal to hardware-based approaches are those \cite{Francillon2014,Cao2023,swatt,pioneer,yan2011,6234416,Armknecht2013} that rely on side-channel information to authenticate the prover and to ensure that prover and verifier run the protocol in a clean environment, i.e. without the interference of a network attacker. These approaches are called \emph{software-based}. 
In between hardware and software-based approach lies the \emph{hybrid} approach, whose goal is either to emulate secure hardware or minimizes its use \cite{Schulz2017,SIMPLE,Carpent2018,FENG2018167,Schulz2011,6881436}. 

 
\noindent \emph{Problem statement.}
In classical computing, it is still an open problem how to achieve remote memory attestation in the presence of a network attacker without assuming trusted hardware on the remote device \cite{Francillon2014,Trujillo-Rasua2019,Armknecht2013}. That is to say, existing software-based attestation protocols whose participants are modeled as Turing machines do not 
satisfy (cryptographic) authentication nor resist  \emph{proxy attacks}: an attack whereby the prover outsources the attestation task to a third-device. 
That is an important 
limitation, which forces the verifier to run the protocol in a clean 
environment aided by the use of out-of-band channels, such as visual 
inspection. 
In quantum computing, however, the problem has been claimed to be successfully addressed by the remote memory attestation protocol \soteria{} \cite{Soteria}, which does assume a Quantum Physically Uncloneable Function (QPUF) on the remote device, but do not require trusted hardware to prevent malware from interfering with the attestation routine. 

\noindent 
\emph{Contributions.} 
This article makes three contributions. 
\begin{enumerate}
    \item First, we show that \soteria{} does not achieve its intended security goals, thereby leaving open the problem of remote attestation without secure hardware in the presence of a network attacker, even within the realm of quantum computing. 
    The flaws and pitfalls we report on \soteria{} are non-trivial, some of them arising from counter-intuitive effects of quantum mechanics. Therefore, we believe they are useful for a more general audience, including, but not limited to, the community working on quantum-based security protocols. 
    
    \item We prove an impossibility result on the use hashing for quantum memory, which is a standard practice in the classical memory attestation problem. In addition, we show a counterintuitive result on the attestation of quantum memory 
    that forces any protocol (where the prover sends its memory content back to verifier) to use a number of messages that quadratically increases with the inverse of the size of the malware. Those two proofs indicate that approaches to the attestation of quantum memory might need to be fundamentally different to existing approaches for classical memory.
    
    \item On a more positive note, we show that quantum computing do offer features that have no counterpart in classical computing, and that can enhance the security of remote attestation protocols. One of those features is \emph{teleportation}, which is used in \cite{LaeuchliT24} to prevent a remote device from trivially outsourcing the attestation routing onto a colluding (more powerful) device. The other feature is \emph{superdense coding}, which we show in this article can be used by the verifier to agree with the prover on a random quantum circuit at runtime. This leads to an improvement over the design introduced in \cite{LaeuchliT24} in terms of resistance to proxy attacks. 
    More precisely, we introduce a protocol that is secure against an attacker that colludes with the prover from a sufficiently long distance. To the best our knowledge, this is the highest level of security remote attestation protocols have achieved so far, in both classical and quantum computing, without resorting into trusted hardware.    
\end{enumerate}

\noindent 
\emph{Organization.} 
In Section \ref{sec-preliminaries}, we review the fundamentals of Quantum Information Science (QIS) necessary to understand out theoretical results. In Section \ref{sec-soteria}, we discuss how \cite{Soteria} approached the problem of remote memory attestation. Since this design illustrates many of the issues that may derail proposed Quantum implementations of memory attestation, we note the correctness and security flaws of this method here, proving several theorems that will impact any potential approach. 
Limitations on the attestation of quantum memory that go beyond \soteria{} are reported in 
Section \ref{sec-quantum-memory}. 
Section \ref{sec-superdense} discuss the proposal of \cite{LaeuchliT24}, and taking into account some of the limitations and flaws reported in previous sections, we design a new system which improves on both approaches. We note that a brief review of the literature on remote memory attestation is provided in Section \ref{sec-introduction}, while details on the two existing quantum-based memory attestation protocols are given elsewhere in the paper.




\section{Preliminaries}
	\label{sec-preliminaries}
	Here we introduce the notation and tools we intend to make use throughout the article. We begin by reviewing the basics of quantum states and quantum 
	teleportation. For a complete introduction to these ideas, the reader is 
	recommended to consult \cite{teleportationiscool}. 
	
	\subsection{Quantum Notation}
	
    We represent qubits in the standard Bra-ket notation, being $\ket{0}$ and $\ket{1}$ the \emph{computational basis states} of a qubit. Unlike bits in classical computing, which are either $0$ or $1$, a qubit can be in superposition of its basis states, denoted by,
	
	\begin{equation}
		\ket{ \psi} = \alpha_0 \ket{0} + \alpha_1 \ket{1}
	\end{equation}
	
    In the notation above, $\psi$ is the \emph{name} of the qubit, while $\alpha_0$ and $\alpha_1$ are complex numbers satisfying $| \alpha_0 |^2+| \alpha_1 
	|^2 = 1$. 
    Because the values 
	$\alpha_0$ and $\alpha_1$ are complex, a qubit lies on a unit sphere known as 
	the Bloch Sphere. Considering $(\alpha_0,\alpha_1)$ a vector on the 
	Bloch sphere with polar angle $\theta$, and azimuthal angle $\phi$, we can 
	represent any qubit as $cos(\frac{\theta}{2}) \ket{1} + 
	sin(\frac{\theta}{2}) e^{i\phi}\ket{0}$.

    We define an $n$-qubit state by,
	
	\begin{equation}
		\ket{\psi} = \sum_{i=0}^{2^n-1} \alpha_i \ket{i} 
	\end{equation}
	
	Where $\alpha_0, \ldots, \alpha_{2^n-1} \in \mathbb{C}$, $\sum^{2^n-1}_{i=0} | \alpha_i |^2 = 
	1$, and, for every $i \in \{0, \ldots, 2^n-1\}$, $\ket{i}$ denotes the combination of the basis states $\ket{b_0} \cdots \ket{b_{n-1}} = \ket{b_0 \cdots b_{n-1}}$ where $b_0 \cdots b_{n-1}$ is $i$'s binary representation. 
    According to the Born rule, the probability of obtaining any basis state $\ket{i}$ on measurement on the computational basis
	is $|\alpha_i|^2$. That is, measuring a quantum state in superposition will make it collapse into one of its basis states. 
    
    At this point it is worth introducing the \emph{Bell basis}, which we will use later to describe quantum teleportation. 
    The Bell basis for 2-qubit systems is represented, in terms of the computational basis, in the following way.   

    \begin{align*}
		\ket{ \Phi^+} &= \frac{1}{\sqrt{2}} \ket{00} + 
		\frac{1}{\sqrt{2}} \ket{11} \\
		\ket{ \Phi^-} &= \frac{1}{\sqrt{2}} \ket{00} - 
		\frac{1}{\sqrt{2}} \ket{11} \\
		\ket{ \Psi^+} &= \frac{1}{\sqrt{2}} \ket{01} + 
		\frac{1}{\sqrt{2}} \ket{10} \\
		\ket{ \Psi^-} &= \frac{1}{\sqrt{2}} \ket{01} - 
		\frac{1}{\sqrt{2}} \ket{10} 
	\end{align*}

    We can represent a 2-qubit state $\ket{\psi}$ in the Bell basis by  writing $\alpha_0\ket{ \Phi^+} + \alpha_1\ket{ \Phi^-} + \alpha_2\ket{ \Psi^+} + \alpha_3\ket{ \Psi^-}$. And, when measuring $\ket{\psi}$ on the Bell basis, we would obtain the basis state $\ket{ \Phi^+}$ with probability $\alpha_0^2$, and so on. If, instead, we measure $\ket{\psi}$ on the computational basis, we would obtain $\ket{00}$ with probability $\frac{(\alpha_0 + \alpha_1)^2}{2}$, and so on. 
    Again, measuring on a given basis, say the computational basis or the Bell basis, allows us to collapse a quantum state into one of the corresponding basis states. 
    
    \subsection{Entanglement}

    We say a state $\ket{\psi}$ is \emph{separable} if it can be written as the tensor 
	product of other states, namely $\ket{\psi} = \ket{\psi_1} \otimes \ket{\psi_2}$.
	States which are not separable are \emph{entangled}. As an example, $\alpha_1 \alpha_2\ket{01}$ is separable because it can be written as $\alpha_1 \ket{0} \otimes \alpha_2 \ket{1}$, whereas all Bell basis states are entangled. 
	
    A key property of entangled two-qubit states
    is that, if we provide one half of the entangled pair 
	to one party, and the other half to another, if either party measures 
	$\ket{0}$ from their half of the pair, the other party will also measure 
	$\ket{0}$, and similarly for $\ket{1}$. 
    For example, 
    if we have a two-qubit state $\frac{1}{\sqrt{2}} \ket{0}_a\ket{0}_b + \frac{1}{\sqrt{2}} \ket{1}_a\ket{1}_b$, where $\ket{0}_a$ (resp. $\ket{1}_b$) denotes a qubit held by a quantum-enabled computational device $a$ (resp. $b$), then when $a$ measures on the computational basis and obtains $\ket{0}$, which occurs with probability one-half, so does $b$, and vice-versa. 
    This property of entanglement is 
	the basis for the power of quantum computing in general. 
 
    Notice how in the previous paragraph we wrote $\frac{1}{\sqrt{2}} \ket{0}_a\ket{0}_b + \frac{1}{\sqrt{2}} \ket{1}_a\ket{1}_b$ instead of $\frac{1}{\sqrt{2}} \ket{00} + \frac{1}{\sqrt{2}} \ket{11}$. Those two-qubit quantum states are equivalent in terms of their behavior; the subscripts are merely used to make it explicit that 
    $\ket{0}_a$ and $\ket{0}_b$ are two different particles, possibly held by two different devices. 


	\subsection{Quantum Teleportation}
	
    Quantum Teleportation is a method for 
	teleporting a qubit from one party to another. The state to be teleported 
	is transmitted instantly over any distance, but cannot be extracted without 
	two classics bits, which can of course be transmitted no faster than the 
	speed of light.
	
    Suppose Alice wishes to transmit a qubit $\ket{\psi}_{c} = 
	\alpha_0 \ket{0}_{c} + \alpha_1 \ket{1}_{c}$ to Bob. Here, the subscript $c$ is used to distinguish the particle to be transmitted from two other particles necessary for teleportation.
    Those two other particles are assumed to be entangled in any of the Bell basis states, say $\ket{ \Phi^+}_{ab} = \frac{1}{\sqrt{2}} \ket{0}_a \ket{0}_b + 
    \frac{1}{\sqrt{2}} \ket{1}_a \ket{1}_b$, where the subscript $a$ (resp. $b$) denotes a particle $a$ (resp. $b$) held 
    by Alice (resp. Bob).  
    The three-qubit system state formed by  $\ket{\psi}_{c}$ and $\ket{\Phi^+}_{ab}$ is thus, 
	
	\begin{align*}
		& \ket{\psi}_c \otimes \ket{\Phi^+}_{ab} = \\ 
        & \quad (\alpha_1 \ket{0}_c + \alpha_2 \ket{1}_c) \otimes 
		(\frac{1}{\sqrt{2}} \ket{0}_a \ket{0}_b + \frac{1}{\sqrt{2}} \ket{1}_a 
		\ket{1}_b ) \\
        & \quad =  \frac{\alpha_1}{\sqrt{2}} \ket{0}_c \ket{0}_a \ket{0}_b + \frac{\alpha_1}{\sqrt{2}} \ket{0}_c \ket{1}_a \ket{1}_b + \\
        & \quad \quad \frac{\alpha_2}{\sqrt{2}} \ket{1}_c \ket{0}_a \ket{0}_b + \frac{\alpha_2}{\sqrt{2}} \ket{1}_c \ket{1}_a \ket{1}_b\\
        & \quad = \frac{\alpha_1}{\sqrt{2}} \ket{00}_{ca} \ket{0}_b + \frac{\alpha_1}{\sqrt{2}} \ket{01}_{ca} \ket{1}_b + \\ 
        & \quad \quad \frac{\alpha_2}{\sqrt{2}} \ket{10}_{ca} \ket{0}_b + \frac{\alpha_2}{\sqrt{2}} \ket{11}_{ca} \ket{1}_b
	\end{align*}
	
    Notice that $\ket{\psi}_c \otimes \ket{\Phi^+}_{ab}$ is written above in the computational basis. The goal next is to rewrite it in the Bell basis. We can do so by using the following identities. 
	
	\begin{align*}
		\ket{00}_{ca} &= \frac{1}{\sqrt{2}}( \ket{\Phi^+}_{ca}+\ket{\Phi^{-1}}_{ca}) \\
		\ket{01}_{ca} &= \frac{1}{\sqrt{2}}( \ket{\Psi^+}_{ca}+\ket{\Psi^{-1}}_{ca}) \\
		\ket{10}_{ca} &= \frac{1}{\sqrt{2}}( \ket{\Psi^+}_{ca}-\ket{\Psi^{-1}}_{ca}) \\
		\ket{11}_{ca} &= \frac{1}{\sqrt{2}}( \ket{\Phi^+}_{ca}-\ket{\Phi^{-1}}_{ca}) \\
	\end{align*}
	
	The total state of the system can thus be rewritten as follows. 
 
	\begin{align*}
		\ket{\psi}_c \otimes \ket{\Phi^+}_{ab} = 
		\frac{1}{\sqrt{2}} \left( \right. & \ket{\Phi^+}_{ca} \otimes(\alpha_1 
		\ket{0}_b+\alpha_2\ket{1}_b)+ \\
		&\ket{\Phi^-}_{ca} \otimes(\alpha_1 \ket{0}_b-\alpha_2\ket{1}_b)+\\
		&\ket{\Psi^+}_{ca} \otimes(\alpha_1 \ket{1}_b+\alpha_2\ket{0}_b)+\\
		&\left. \ket{\Psi^-}_{ca} \otimes(\alpha_1 \ket{1}_b-\alpha_2\ket{0}_b) \right)\\
	\end{align*}
	
    Now Alice can measure her two qubits $a$ and $c$ in the Bell basis. 
    This leads, with equal probability, the system into one of the following states.
	
	\begin{align*}
		&\ket{\Phi^+}_{ca} \otimes(\alpha_1 \ket{0}_b+\alpha_2\ket{1}_b) \\
		&\ket{\Phi^-}_{ca} \otimes(\alpha_1 \ket{0}_b-\alpha_2\ket{1}_b)\\
		&\ket{\Psi^+}_{ca} \otimes(\alpha_1 \ket{1}_b+\alpha_2\ket{0}_b)\\
		&\ket{\Psi^-}_{ca} \otimes(\alpha_1 \ket{1}_b-\alpha_2\ket{0}_b)\\
		\label{equ:4states}
	\end{align*}

	Alice then inform Bob which of the four states the total system is in, 
	based on the result of her measurement by transmitting 2 classical bits. 
    If 
	the system is in the first state, Bob does nothing, otherwise he performs a 
	rotation on his qubit to return it to the form $\alpha_1 
	\ket{0}_b+\alpha_2\ket{1}_b$. Thus, as long as Alice can transmit two 
	classical bits and share a Bell basis state, she can teleport an arbitrary quantum state to Bob. 


    \subsection{Superdense Coding}

    While quantum teleportation allows us to transmit a qubit using two bits of classical information, \emph{superdense coding} allows us to transmit two classical bits of information by transmitting one pair of an entangled qubit \cite{superdense}. If the adversary obtains either of the entangled qubits separately, this provides no information about which classical pair of bits is being transmitted. The recipient, and thereby any man-in-the-middle adversary, must have both qubits in hand in order to obtain the classical information.

    Like in teleportation, the protocol starts with Alice and Bob sharing a Bell state, say 
    $ \ket{ \Phi^+}_{ab} = \frac{1}{\sqrt{2}} \ket{0}_a \ket{0}_b + 
		\frac{1}{\sqrt{2}} \ket{1}_a \ket{1}_b$.
    Once Alice decides on the two bits she wishes to send to Bob, denoted $b_0b_1$, she rotates her qubit in such a way that the state $\ket{ \Phi^+}_{ab}$ is transformed into the $i$th Bell basis state where $i$ is the integer whose binary represeantion is $b_0b_1$. Notably, if $b_0b_1 = 00$, no rotation is performed and $\ket{ \Phi^+}_{ab}$ remains as initially shared.    
 After performing the rotation (if necessary), Alice sends her qubit to Bob. Upon reception of Alice's qubit, Bob will apply the Controlled NOT gate, with Alice's qubit as the control qubit, and his original qubit as the target qubit. Afterwards, he will apply the Hadamard gate on Alice's qubit. On the completion of this process he will measure either $\ket{00}, \ket{01}, \ket{10}$, or $\ket{11}$, depending on which Bell state Alice created before sending her qubit to Bob. 
 Alice has thus transferred two bits of classical information, by transmitting one qubit.

    \subsection{Amplitude encoding and quantum tomography}

    Amplitude enconding is a mechanism to represent classical information within a quantum state. Given a vector $X = (x_1, \ldots, x_n)$ of reals, amplitude enconding 
    codifies $X$ in a quantum state $\ket{\psi_X}$ as follows.  

    \begin{align*}
        \frac{1}{N} \left[ x_1\ket{0} + \cdots + x_{n} \ket{n-1}\right] 
    \end{align*}

    Where $N = \sqrt{x_1^2 + \cdots + x_{n}^2}$ is called the \emph{normalizing factor} of the enconding, which is necessary to ensure that the sum of squares of the amplitudes of the quantum state equals one.  

    The decoding process, that of extracting $X$ out of $\ket{\psi_X}$, is known by \emph{quantum tomography} \cite{tomography}. Quantum tomography measures a series of the same qubit over and over to attempt to extract the value of the qubit. 
    Due to the probabilistic nature of measurement, the convergence rate of quantum tomography for an unknown qubit requires $O(\frac{n}{\epsilon^2})$ number of samples, 
    where $\epsilon$ is the desired error. 


   

    \section{Flaws and pitfalls on \soteria}
    \label{sec-soteria}
    This section examines \soteria{} and reports on its flaws and pitfalls. We shall introduce the two versions of \soteria, one addressing classical memory and the other one quantum memory. The analysis here focuses on the classical version, though, as the analysis of the quantum version is subsumed by the more general analysis we give in the next section.  
    




    \subsection{Specification of \soteria{}}

    The \soteria{} protocol aims at allowing a verifier to attest the memory of an IoT device, called a prover. 
    It does so by assuming that both, prover and verifier, can perform quantum computation. In particular, the prover is manufactured with an \emph{ideal} Quantum Physically Uncloneable Function (QPUF). 
    
    A QPUF has domain and range the universe of quantum states. A QPUF is termed \emph{ideal} if we assume that, for every input quantum state $\ket{\psi}$, its output cannot be predicted unless the QPUF has been previously called on $\ket{\psi}$. In other words, an ideal QPUF behaves like an ideal hash function (modelled as a random oracle) where each possible query is mapped to a (fixed) random response from its output domain.

    The last piece of the prover's architecture necessary to understand \soteria{} is its memory, which is split into $m$ words of memory\footnote{The original article assumes words of memory to be grouped into blocks. Because they play no crucial role in the security or efficiency of the protocol, we have ommited the notion of blocks. }. 
    Notably, \soteria{} can be adapted to work on both quantum and classical memory. In the classical memory setting, an index is an integer in $\{0, \ldots, m-1\}$ that  refers to the position of the memory word, and a word is a binary string of fixed length, typically 16 or 32 bits.   
    In the quantum memory setting, an index is a quantum state $\alpha_0 \ket{0} + \cdots \alpha_{m-1} \ket{m-1}$ (possibly) in superposition, and so is a word. Let $\{\ket{\psi_0}, \ldots, \ket{\psi_{m-1}}\}$ be the ordered set of words in the quantum memory.   
    Then, given an index  $\alpha_0 \ket{0} + \cdots \alpha_{m-1} \ket{m-1}$, the memory content is calculated by $\alpha_0 \ket{\psi_0} + \cdots \alpha_{m-1} \ket{\psi_{m-1}}$. In particular, if $\alpha_i = 1$, then the index gives exactly the word $\ket{\psi_i}$ rather than the superposition of many words.

    The \soteria{} protocol splits into the following four phases. 

    \noindent \emph{Setup phase.} 
    During setup, 
    the verifer generates a random bitstring $c$ of length $\kappa$, where $\kappa$ is a security parameter, and uses amplitude encoding to create a  
    quantum state $\ket{\psi_c}$ (also of length $\kappa$).     
    The verifier then queries the prover's QPUF on input $\ket{\psi_c}$ to obtain $\ket{\psi_r}$.  
    This process is performed for a sufficiently large number of random inputs. Let $\mathcal{I}_P$ denote all pairs $(\ket{\psi_c}, \ket{\psi_r})$ stored during the setup phase between the verifier and a prover $P$. 

    \noindent \emph{Challenge phase.} 
    To attest the memory of a prover $P$, the verifier randomly picks a pair $(\ket{\psi_c}, \ket{\psi_r})$ out of $\mathcal{I}_P$ which shall determine the words of memory to be attested (see Figure \ref{fig-soteria-quantum} and \ref{fig-soteria-classical}). 
    In the classical memory setting, the verifier also    
    generates a nonce $R_w$\footnote{The original article uses an additional nonce for memory blocks, which isn't necessary here.}. 
    The challenge $\ket{\psi_c}$ is sent to the prover over a quantum channel, while 
    $R_w$ is sent (if necessary) over a classical channel. 
    
    \noindent \emph{Response phase.} The prover calls its QPUF on input $\ket{\psi_c}$ to obtain the corresponding output $\ket{\psi_r}$, as previously agreed with the verifier. The next step depends on the memory setting we are in.
    
    \begin{itemize}
        \item In the quantum memory setting, $\ket{\psi_r}$ is used directly as a memory index. Let $\ket{\psi_{\sigma}}$ be the quantum state corresponding to the memory indexed at $\ket{\psi_r}$. 
        \item In the classical memory setting, the prover extracts a bitstring $r$ of length $\kappa$ from $\ket{\psi_r}$. 
        \soteria{} then follows the approach \cite{Hatt} to compute a checksum $\sigma$ of the memory by using $r \oplus R_w$ as seed of a pseudo-random function to randomly sample memory indexes.         
 Lastly, the prover uses amplitude enconding to encode $\sigma$ into a quantum state $\ket{\psi_{\sigma}}$, which is sent to the verifier. 
    \end{itemize}

    \noindent \emph{Verification phase.}
    Upon reception of the quantum state $\ket{\psi_{\sigma}}$, the verifier goes through the same steps as the prover to obtain its own attestation value of the provers memory based on the challenge $\ket{\psi_c}$. 
    Note that the verifier has all the information to calculate a state equal to $\ket{\psi_{\sigma}}$, including the content of the prover's memory and the challenge-reponse pairs obtained from the QPUF. 
    Let $\ket{\psi_{\sigma_P}}$ be the quantum state obtained by the verifier.
    The verifier then proceeds to check whether $\ket{\psi_{\sigma}} = \ket{\psi_{\sigma_P}}$. 
    If the verifier cannot obtain the same quantum state sent by the prover, then the verifier halts the protocol and raises an alarm. 
    
    \begin{figure}
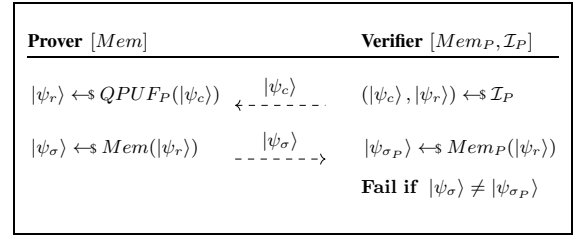

		\begin{center}
  \scalebox{0.8}{
			\fbox{
				\pseudocode[colspace=0cm]{%
					\< \< 
					\\
					\textbf{Prover $[Mem] $} \< \< 
					\textbf{Verifier 
						$[Mem_P, \mathcal{I}_P]$}  
					\\[0.1\baselineskip][\hline]
					\<\< \\[-0.5\baselineskip]
					\ket{\psi_r} \sample QPUF_P(\ket{\psi_c}) \< \sendmessageleft*[1.5cm, style ={dashed}]{\ket{\psi_c}}  \< (\ket{\psi_c}, \ket{\psi_r}) \sample \mathcal{I}_P\\
					\ket{\psi_{\sigma}} \sample Mem(\ket{\psi_r}) \<		\sendmessageright*[1.5cm, style={dashed}]{\ket{\psi_{\sigma}}} \<
                    \ket{\psi_{\sigma_P}} \sample Mem_P(\ket{\psi_r}) 
                    \\
					\<\< \highlightkeyword{Fail} \highlightkeyword{if} 
                    \ket{\psi_{\sigma}} \neq \ket{\psi_{\sigma_P}}
					\\
				} 
			}
   }
			\caption{\soteria{} for quantum memory where $Mem$ denotes the memory content of the prover and $Mem(\ket{\psi_r})$ denotes the memory value at index $\ket{\psi_r}$.  
				\label{fig-soteria-quantum}}	
		\end{center}
	\end{figure}

    \begin{figure}
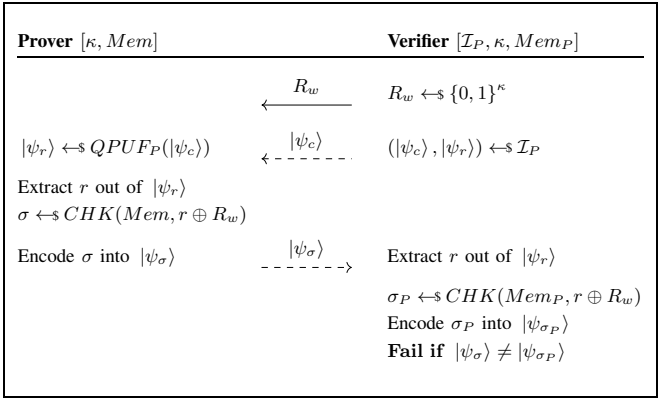

		\begin{center}
  \scalebox{0.8}{
			\fbox{
				\pseudocode[colspace=0cm]{%
					\< \< 
					\\
					\textbf{Prover $[\kappa, Mem]$} \< \< 
					\textbf{Verifier 
						$[\mathcal{I}_P, \kappa, Mem_P]$}  
					\\[0.1\baselineskip][\hline]
					\<\< \\[-0.5\baselineskip]
					\< \sendmessageleft*[1.5cm]{R_w}  \< R_w \sample \bin^{\kappa} \\
					\ket{\psi_r} \sample QPUF_P(\ket{\psi_c}) \< \sendmessageleft*[1.5cm, style ={dashed}]{\ket{\psi_c}}   \< (\ket{\psi_c}, \ket{\psi_r}) \sample \mathcal{I}_P\\
					\text{Extract } r \text{ out of } \ket{\psi_r}\< \<  
					\\
					\sigma \sample CHK(Mem, r \oplus R_w) \<\<  
					\\
					\text{Encode } \sigma \text{ into } \ket{\psi_{\sigma}} \< 
					\sendmessageright*[1.5cm, style={dashed}]{\ket{\psi_{\sigma}}} \<
					\text{Extract } r \text{ out of } \ket{\psi_r}\\
					\<\<  
					\sigma_P \sample CHK(Mem_P, r \oplus R_w) \\
					 \< 
					\< \text{Encode } \sigma_P \text{ into } \ket{\psi_{\sigma_P}} 
					\\
					\<\< \highlightkeyword{Fail} \highlightkeyword{if} 
                    \ket{\psi_{\sigma}} \neq \ket{\psi_{\sigma_P}}
					\\
				} 
			}
   }
			\caption{\soteria{} for classical memory where $CHK(Mem, r \oplus R_w)$ is a (possibly partial) checksum on the prover's memory $Mem$ with seed $r \oplus R_w$. Solid and solid arrows are used to denote classical and quantum communication, respectively. 
				\label{fig-soteria-classical}}	
		\end{center}
	\end{figure}

   \subsection{Correctness analysis}


   Correctness in memory attestation means that, if the prover has not been corrupted, then it should pass the protocol with a  verifier. 
   In the case of \soteria, and that of most memory attestation protocols, the prover can be modelled as a probabilistic algorithm $P$ with domain the universe of challenges $\mathcal{C}$ and range the universe of responses $\mathcal{R}$. 
   The verifier, instead, is a probabilistic algorithm $V$ that, on input $P, c \in \mathcal{C}, r \in \mathcal{R}$, outputs 
   \emph{pass} if 
   $P(c) = r$, \emph{fail} otherwise. 
   This leads to the following definition of correctness.  

    \begin{definition}[Correctness]
        A memory attestation protocol is \emph{correct} if for every honest prover $P$, honest verifier $V$ and challenge $c \in \mathcal{C}$, the probability of $V(P, c, P(c)) = pass$ is equal to $1$. 
    \end{definition}

    Notice that correctness assume neither noise nor adversarial interference in the protocol, making correctness a minimum requirement for attestation protocols. 
    Correctness is also a necessary condition for a meaningful security analysis, as a protocol that rejects all provers has a $100\%$ detection rate.   

    We show next that either \soteria{}  
    is incorrect or it relies in a $qPUF$ that is predictable hence insecure. 

    \begin{lemma}
        The classical version of \soteria{} either does not satisfy correctness or uses a qPUF implementation that outputs a basis state for every input.  
    \end{lemma}

    \newcommand{\innerproduct}[2]{\left\langle #1, #2 \right\rangle}

    \begin{proof}

        In \soteria, the prover receives a single challange $\ket{\psi_c}$ from the verifier. 
        When the prover calls 
        $\ket{\psi_r} = qPUF(\ket{\psi_c})$, the quantum state 
        $\ket{\psi_c}$ collapses, meaning that $r$ should be extracted out of a single copy of $\ket{\psi_r}$. 
        
        Suppose 
        $\ket{\psi_r} = \alpha_0\ket{0} + \cdots + \alpha_{\kappa}\ket{\kappa}$, where $\kappa$ is the security parameter. 
        Now, after measuring $\ket{\psi_r}$, the prover will obtain $\ket{i}$ with probability $\alpha_i^2$. Likewise, the verifier will obtain $\ket{j}$ with probability $\alpha_j^2$. 
        This means that the probability of the prover and verifier obtaining the same basis state is $\sum_{i = 0}^{i = \kappa} \alpha_i^4$. Such a probability is equal to $1$ only when $\alpha_i = 1$ for some $i \in \{0, \ldots, \kappa\}$. 
        Therefore, either the qPUF outputs a predictable 
        basis state $\ket{i}$, which contradicts the security properties of qPUFs, 
        or the 
        protocol does not satisfy correctness. 
    \end{proof}

    \subsection{Analyzing potential fixes for correctness}

    Based on the proof above, naively it appears that \soteria{} can be fixed by letting the 
    verifier sends many copies of the challenge state.
    We show however, that the classical version of \soteria{} requires at least a quadratic number of messages in terms of the security parameter $\kappa$, rendering the protocol unviable for reasonable security values. 

    \begin{theorem}
        Consider a version of \soteria{} where the classical prover is able to correctly decode the bitstring $r$. That version requires $\mathcal{O}(\kappa^2)$ copies of the challenge $\ket{\psi_c}$. 
    \end{theorem}
    \begin{proof}
        As mentioned in the preliminaries, the only way  
        to extract the amplitudes from a qubit is the process of quantum tomography, which requires measuring the same qubit over and over to estimate the amplitudes value of the qubit. 
        The question then is how many measurements \soteria{} needs to extract $r$ out of $\ket{\psi_r}$. 

        Suppose we perform quantum tomography on $\ket{\psi_r}$ up to accuracy $\epsilon$, i.e. $\epsilon$ is the maximum distance from the original amplitudes $\{ \alpha_1, \ldots, \alpha_{\kappa}\}$ and the estimate amplitudes $\{\alpha_1', \ldots, \alpha_{\kappa'}\}$, which is to say that $\sum|\alpha_i-\alpha_i'| = \epsilon$.
        The next step is to map the vector of the real values $\{ \alpha_1, \ldots, \alpha_{\kappa}\}$ into a bitstring $r = r_1 \cdots r_{\kappa}$. 
        \soteria{} doesn't give a procedure to make such mapping. However, because they  
        use amplitude encoding for transforming bistring into state amplitudes, it is fair to assume that they divide by the normalization factor to recover the bitstring out of the amplitudes of a quantum state. 
        That is to say, there exists a normalization factor $N$ known by both the prover and the verifier satisfying that $r_i = N\alpha_i $ with $N = \sqrt{r_1^2 + \cdots r_{\kappa}^2}$, which would allow the devices to obtain $r_1, \ldots, r_{\kappa}$ out of the amplitudes of $\ket{\psi_r}$. 

        Now, let $\{\alpha_1^p, \ldots, \alpha_{\kappa}^p\}$ and $\{\alpha_1^v, \ldots, \alpha_{\kappa}^v\}$ be the estimated amplitudes obtained by the prover and the verifier, respectively. Then, the memory indexes calculated by the prover and the verifier are, respectively, $(r_1^p, \ldots, r_{\kappa}^p) = (\lfloor N \alpha_1^p \rceil, \ldots, \lfloor N \alpha_{\kappa}^p \rceil)$ and $(r_1^v, \ldots, r_{\kappa}^v) = (\lfloor N \alpha_1^v \rceil, \ldots, \lfloor N \alpha_{\kappa}^v \rceil)$. 
        Because $r$ is used as a seed of a pseudo-random function, \soteria{} needs $r_1^p \cdots r_{\kappa}^p = r_1^v \ldots r_{\kappa}^v$ to function properly. 
        
        Recall that $\alpha_i^p = \alpha_i + \epsilon_i^p$ for every $i$ and that 
        $\sum |\epsilon_i^p| = \epsilon$. Likewise, considering that the verifier would estimate the amplitudes up to the same accuracy as the prover, we obtain that $\alpha_i^v = \alpha_i + \epsilon_i^v$ with  
        $\sum |\epsilon_i^v| = \epsilon$. So  
        the worst case for ensuring that $r_i^p = r_i^v$ is when 
        all the error is concentrated on the $i$th component for, say the prover, while there is no error on the $i$th component of the verifier. That is, when $\epsilon_i^v = 0$ and $\epsilon_i^p = \epsilon$ or vice versa. In that case,  
        we have $r_i^p = r_i^v \iff \lfloor N\alpha_i \rceil = \lfloor N(\alpha_i + \epsilon) \rceil$, which implies that \soteria{} needs $N\epsilon < 0.5$, otherwise the protocol would incorrectly believe that the prover has been compromised. 

        Once we have established that $\epsilon < \frac{1}{2N}$, the next proof step is to use the   
        result in \cite{tomography}, which states that the convergence rate of quantum tomography for an unknown qubit $\ket{\psi_r}$ 
        requires $\mathcal{O}(\frac{\kappa}{\epsilon^2})$ number of samples, where $\epsilon$ is the desired error.
        That is to say, \soteria{} requires $\mathcal{O}(\kappa N^2)$ measurements of $\ket{\psi_r}$.  
        Due to the no-cloning theorem, we cannot make additional copies of $\ket{\psi_r}$, we must produce new ones. To do this \soteria{} will need new copies of $\ket{\psi_c}$ from the verifier. In other words, the prover in \soteria{}  requires a quadratic number of copies from the verifier.  
        
        We end the proof by noticing that, because $r_i$ is a bit, $N$ is bounded by $\sqrt{\kappa}$, which gives a communication complexity $\mathcal{O}(\kappa^2)$ for the response phase of \soteria.  
    \end{proof}

    We remark that, if there is prior knowledge of the value of the qubit to be decoded, then one can use Bayesian approaches to attempt to speed the convergence of this method. However, the output of the QPUF must be unknown to the prover (and thus also to the attacker). Thus these methods cannot be employed.
    
    \subsection{Security analysis assuming an isolated execution environment}

    A security analysis requires an adversary model. The standard adversary for software-based memory attestation 
    consists of a malicious prover with corrupt memory trying to convince the verifer that its memory content is in a safe state. 
    Notably, in this model, the prover receives no help from a third-device or colluder, known as the \emph{isolation assumption}. 
    
    We model the adversary based on a security experiment $\textsc{Exp}^{\adv}_{P, d, \delta}$ where $d$ is an unespecified distance function on memory content and $\delta$ a distance threshold (see Figure \ref{fig-sec-exp}). 
    In the setup phase, the adversary is given 
    a memory content $Mem_{P'}$ such that $d(Mem_{P'}, Mem_{P}) < \delta$. By giving the adversary a random memory at a $\delta$-distance from the correct memory  $Mem_P$, we are avoiding attacks whereby the adversary reconstructs $Mem_{P}$ from $Mem_{P'}$. 
    In the challenge phase of the experiment, the adversary outputs a malicious prover $P'$ with memory $Mem_{P'}$. 
    The adversary wins the experiment if, given a random challenge $c \in \mathcal{C}$, it holds that 
    $P'(c) = P(c)$. 

\begin{figure}
\procedureblock[linenumbering]{$\textsc{Exp}^{\adv}_{P, d, \delta}$}{
    Mem_{P'} \sample \{x | d(x, Mem_P) < \delta\}\\
    P' \sample \adv(Mem_{P'})\\
    \ket{\psi_c} \sample \mathcal{C}\\
    \pcreturn P'(\ket{\psi_c}) = P(\ket{\psi_c}) 
    }
\caption{The security experiment. \label{fig-sec-exp}}            
\end{figure}

\begin{definition}[$\epsilon$-secure attestation]
A software-based attestation protocol is said to be $\epsilon$-secure if for every prover $P$ the probability of the adversary winning $\textsc{Exp}^{\adv}_{P}$ is lower than $\epsilon$. 
\end{definition}

The definition above ensures that, whenever a prover has been corrupt in an irreversable manner, i.e. in a way that the corrupt prover cannot recover its original state, then an $\epsilon$-secure attestation protocol will detect the corruption with probability at least $1 - \epsilon$. 
We note this is a fairly conservative definition of security. Nonetheless, one that \soteria{} cannot meet when it is assumed correct, as we show next. 

\begin{theorem}
Assuming the destructive version of the SWAP test, \soteria{} cannot detect attacks with probability higher than $\frac{1}{2}$. 
\end{theorem}
\begin{proof}
For the proof we assume a correct  version of \soteria{}, i.e. one where the prover can correctly decode the bitstring $r$ even at the cost of (at least) a quadratic number of messages. Recall that incorrect protocols can achieve arbitrary levels of security. We also assume the destructive version of the SWAP test, as the non-destructive version requires high depth circuits and extremely low noise levels to keep the states in super position across many runs. This approach is, therefore, not suitable for near term quantum computing. 

The SWAP test \cite{swaptest} takes two input states $\ket{\phi}, \ket{\psi}$ and outputs a Bernoulli random variable that is $0$ with probability $\frac{1}{2} + \frac{1}{2}\innerproduct{\phi}{\psi}^2$, where $\innerproduct{\phi}{\psi}$ is the inner product of the two states. 
If two states $\ket{\psi}$ and $\ket{\phi}$ are identical, the test will return $0$ with probability $1$. 
This means that, to satisfy correctness, the verifier ought to accept responses for which the SWAP tests output $0$. Now, being that the case, then any attacker $P'$ wins the security experiment with probability $\frac{1}{2} + \frac{1}{2}\innerproduct{P'(\ket{\psi_c})}{P(\ket{\psi_c})}$. In the worst case, i.e. when the two states are orthogonal, the attacker still has a $\frac{1}{2}$ chance of passing the protocol.

\end{proof}

To enhance the security of \soteria, one would need to run the SWAP test many times, which forces the prover to create and send many copies of $\ket{\sigma}$. In general, if the square of the inner product of two states differ by $\epsilon$, then $n$ tests will be required to detect an attack with probability $1-(1-(\frac{1}{2}-(\frac{1-\epsilon}{2})))^n$. 
In other words, the verifier needs $O(\frac{1}{\epsilon^2})$ 
responses (copies of the quantum state) from the prover to detect an attack that deviates from the correct response by an $\epsilon$ error. 
This requirement greatly increases the actual number of runs of the proposed algorithm needed to attest the prover's memory beyond what is reported in \cite{Soteria}.

\subsection{Analyzing the impact of the SWAP test on Memory Attestation}
 
    We note another counter-intuitive implication of the SWAP test to this application. In classical communication, the prover has incentive to increase the length of vector of memory bits $\sigma$ sent to the prover. In quantum communication, however, this is not necessarily the case. Consider the example $\ket{\psi}_{\sigma} = \frac{1}{\sqrt{3}} \ket{0} + 0 \ket{1} + \frac{1}{\sqrt{3}} \ket{2}  + \frac{1}{\sqrt{3}} \ket{3}+ 0\ket{4}$, whereas the true memory is encoded on the verifier as $\ket{\psi}_{\sigma^*} = 0 \ket{0} + 0 \ket{1} + \frac{1}{\sqrt{2}} \ket{2}  + \frac{1}{\sqrt{2}} \ket{3}+ 0 \ket{4} $, we have that $\braket{{\psi}_{\sigma} | {\psi}_{\sigma^*}} = 0.8166$, showing we have a $p=\frac{1}{2}-\frac{0.8166^2}{2}$ probability to detect the attack after the swap test. On the other hand if we expand our sampled bits by one, assuming we sample an uncorrupted bit, we obtain $\ket{\psi}_{\sigma} = \frac{1}{\sqrt{4}} \ket{0} + 0 \ket{1} + \frac{1}{\sqrt{4}} \ket{2}  + \frac{1}{\sqrt{4}} \ket{3}+ 0\ket{4}+ \frac{1}{\sqrt{4}}\ket{5}$ , $\ket{\psi}_{\sigma^*} = 0 \ket{0} + 0 \ket{1} + \frac{1}{\sqrt{3}} \ket{2}  + \frac{1}{\sqrt{3}} \ket{3}+ 0 \ket{4} + \frac{1}{\sqrt{3}} \ket{5}$, which yields  $\braket{{\psi}_{\sigma} | {\psi}_{\sigma^*}} = 0.8661$, showing we have a $p=\frac{1}{2}-\frac{0.8661^2}{2}$ probability to detect the attack after a single swap test. Paradoxically, even though we have tested more bits of memory we have a lower chance to detect the attack. 

    We can analyze this protocol to decide what the average probability of detecting the attack is, based on the length of bits sampled. 

    \begin{lemma}
    We show that increasing the number of bits sampled does not on average increase the probability of detecting an attack.  
    \end{lemma}
    \begin{proof}
    Let $b$ be the number of bits sampled from memory. Let $\sigma$ be the sampled memory on the prover and $\sigma_*$ be the sampled memory on the verifier. Let $m$ be the number of positive bits sampled, 
    which in expectation are $m=\bold{E}[\frac{b}{2}]$.
    Let $\epsilon$ be the probability that a bit is flipped, chosen by the attacker. 
    Let $\ket{\psi}_{\sigma} = \sum_{i=0, \sigma_i \neq 0}^{b}\frac{1}{\sqrt{m} } \ket{i} $ be the encoded sampled memory, and let $\ket{\psi}_{\sigma^*}=\sum_{i=0, \sigma_*(i) \neq 0}^{b}\frac{1}{\sqrt{m} } \ket{i}$, be the encoded true memory that the verifier is using to compare with the prover's memory.  
    Let $S$ be the number of bits where the memory matches in both the prover's and verifier's sample. We can compute this as $\bold{E}[S]=m(1-\epsilon)$. We note that in expectation both the prover and verifier will have the same number of positive bits in their sample. 

    We want to compute  $\bold{E}[|\braket{\psi_{\sigma}|\psi_{\sigma^*}}|^2]$. Where either  element of the vector is 0, that element does not contribute to the dot product. We then have $\bold{E}[|\braket{\psi_{\sigma}|\psi_{\sigma^*}}|^2]=|(S(\frac{1}{\sqrt{m}}\frac{1}{\sqrt{m}}))|^2=|m(1-\epsilon)(\frac{1}{m})|^2=|1-\epsilon|^2$. 

    Since this is independent of $m$, increasing the length of bits sampled does not increase the probability of detecting the attack. 
     \end{proof}

\section{On the limitations of quantum memory attestation}
\label{sec-quantum-memory}

In this section we discuss limitations of quantum memory that make the memory attestion problem challenging. These limitations are applicable to all challenge-response protocols aiming at attesting quantum memory, including the quantum version of \soteria{} displayed in Figure \ref{fig-soteria-quantum}. 

 \subsection{Overview of Quantum Memory}

 Quantum memory \cite{memory}, like regular memory, allows for the storage of words of memory in the system, each containing a qubit or collection of qubits. However, unlike classical memory, the system does not have to restrict itself to returning only a single word. A quantum memory system can return a super-position of these words. Given a memory index $\ket{\psi}_R$, the system returns a super-position of each memory word weighted by the amplitudes of $\ket{\psi}_R$. 
 
 Since we can return a super-position of quantum memory words, it might seem like it would be easier to design a remote memory attestation procedure than a classical memory attestation procedure, because one can sample from many words of memory simultaneously. However, our analysis has revealed three issues that in practice any quantum remote memory attestation system would have to overcome that make this a challenging problem.  
 

     
     
\subsection{Issue \#1: Quantum memory cannot be copied}

The no cloning theorem states that in general quantum states cannot be copied \cite{Wootters1982Single}. This is a serious problem for quantum remote memory attestation, because it implies that the state of the quantum memory cannot be copied and distributed to a remote verifier. The memory must moved from the prover to the verifier. This means that either the process is a destructive one from the prover's standpoint, where some amount of memory must be sacrificed to attest the remainder, or the verifier must test the memory in a non-destructive way, then return it to the prover. 

Keeping the memory in super-position, transmitting it to the verifier, then comparing it to another quantum state using the swap test, and then returning it to the prover, all without collapsing the super-position, is not a practical proposal for any near term quantum devices. This suggests that the more practical approach would be for the prover to sacrifice some number of qubits of its quantum memory for the attestation protocol. 
Even in that case, we show next that a large number of qubits will be required to attest the memory, in the precense of an attacker. 
    
 \subsection{Issue \#2: Quadratic communication from the prover required}

    Even if the prover is willing to sacrifice some subset of its memory in order to satisfy a remote attestation protocol, we find that if the attacker keeps the number of changes small,  the prover will need to transmit a larger number of qubits to the verifier. In fact, because the memory elements can be complex numbers, the attacker could theoretically make changes which are undetectable. Consider the case where the attacker makes the prover's memory to hold $\ket{\psi}_{\sigma} = \frac{1}{\sqrt{2}} \ket{0} +\frac{1}{\sqrt{2}} \ket{1}$, where the true memory is $\ket{\psi}_{\sigma} = \frac{i}{\sqrt{2}} \ket{0} +\frac{i}{\sqrt{2}} \ket{1}$. Here we have $|\braket{\psi_{\sigma} | \psi_{\sigma*}}|^2 = 1$, meaning the swap test will never detect a difference between these two vectors.

    The main issue is the comparison between the prover's state and the verifier's state. To detect a difference between the two, we must again resort to the swap test. Again, counter-intuitively our probability of detecting the attack does not increase by increasing by sampling more words of the devices memory. 

    \begin{lemma}
    Suppose that for each qubit of quantum memory the attacker makes a change of $|\epsilon| \leq 1$, with probability $p$. 
    We then show that the probability of detecting the attack is proportional to $O(|\epsilon|^2)$, regardless of the number of qubits sampled. 
\end{lemma}
\begin{proof}

    Let $\ket{\psi}_R$ be a state used to index into the quantum memory. Suppose we have $m$ qubits of memory, $\ket{\rho_i}$, which with probability $p$ has been modified to be $\ket{\rho_i}+\epsilon$. 
    After indexing 
 on the prover we have $\ket{\psi}_{\sigma} = \ket{\psi}_R  \otimes (\ket{\rho_i}+\epsilon..)$, while on the verifier we have $\ket{\psi}_{\sigma*} = \ket{\psi}_R  \otimes (\ket{\rho_i}..)$. Thus state $\ket{\psi}_{\sigma} = [\ket{\psi}_{R_{1}} \otimes  (\ket{\rho_1}+\epsilon),\ket{\psi}_{R_{2}} \otimes (\ket{\rho_2}+\epsilon),   ...]$ and  $\ket{\psi}_{\sigma*} = [\ket{\psi}_{R_{1}} \otimes (\ket{\rho_1}),\ket{\psi}_{R_{2}} \otimes (\ket{\rho_2}),   ...]$. 
 
 We have 
 \begin{align*}
    &E[|\ket{\psi}_{\sigma}-\ket{\psi}_{\sigma*}|_2]=\\
      &|[\ket{\psi}_{R_{1}} \otimes (\ket{\rho_1}+\epsilon),\ket{\psi}_{R_{2}} \otimes (\ket{\rho_2}+\epsilon),   ...]-\\ 
      &[\ket{\psi}_{R_{1}}  \otimes(\ket{\rho_1}),\ket{\psi}_{R_{2}} \otimes (\ket{\rho_2}),   ...]|_2 =\\
      &|[\ket{\psi}_{R_{1}}p\epsilon,\ket{\psi}_{R_{2}}p\epsilon,..]|_2=\\
      &\sqrt{\sum_{i=0}^m |\ket{\psi}_{R_{i}}p\epsilon|^2} = p\epsilon
  \end{align*}

  This implies
    \begin{align*}
    & \sqrt{ |\ket{\psi}_{\sigma}|^2 + |\ket{\psi}_{\sigma*}|^2 - 2 |\braket{{\psi}_{\sigma} | {\psi}_{\sigma*}}|} \leq |\ket{\psi}_{\sigma}-\ket{\psi}_{\sigma*}|_2
			\\
      & \sqrt{ |\ket{\psi}_{\sigma}|^2 + |\ket{\psi}_{\sigma*}|^2 - 2 |\braket{{\psi}_{\sigma} | {\psi}_{\sigma*}}|} \leq  
      \\
       &|\ket{\psi}_R  \otimes (\ket{\rho_i}+\epsilon,...)-\ket{\psi}_R  \otimes (\ket{\rho_i},...)|_2
			\\
       & \sqrt{ |\ket{\psi}_{\sigma}|^2 + |\ket{\psi}_{\sigma*}|^2 - 2 |\braket{{\psi}_{\sigma} | {\psi}_{\sigma*}}|} \leq |p\epsilon|_2
			\\
		&   |1|^2 + |1|^2 - 2 |\braket{{\psi}_{\sigma} | {\psi}_{\sigma*}} | \leq 	|p\epsilon|^2_2\\
            & 2 |\braket{{\psi}_{\sigma} | {\psi}_{\sigma*}} |\geq 	2 - 	|p\epsilon|^2_2 \\
            & 4| \braket{{\psi}_{\sigma} | {\psi}_{\sigma*}}|^2 \geq  	4 - 4|p\epsilon|^2_2 +|p\epsilon|^4_2 \\
            & 4| \braket{{\psi}_{\sigma} | {\psi}_{\sigma*}}|^2 - 4 \geq  - 4|p\epsilon|^2_2 +|p\epsilon|^4_2 \\
         &   4 - 4| \braket{{\psi}_{\sigma} | {\psi}_{\sigma*}}|^2 \leq   4|p\epsilon|^2_2 -|p\epsilon|^4_2 \\  
         &   \frac{1}{2} - \frac{| \braket{{\psi}_{\sigma} | {\psi}_{\sigma*}}|^2}{2} \leq   \frac{|p\epsilon|^2_2}{2} -\frac{|p\epsilon|^4_2}{8} \\  
    \end{align*}

    \end{proof}

    We see then that if the attacker keeps $|\epsilon|$  small then the number of qubits needed to detect the attack with high probability may overwhelm the communication channel between the prover and the verifier.  We also observe that as in the case of classical memory increasing the number of words that the device indexes does not increase the chance to detect the attack, a very counter intuitive result .

     \subsection{Issue \#3: Impossibility of hashing quantum states}

    To avoid the problem of an attacker being able to force a large number of rounds needed to detect a change in the memory, we might hope to design a hash of quantum states (this is distinct from a quantum hash of classical states, algorithms for which do exist).  
    That is, we could attempt to design a quantum circuit $U$, such that given two qubits containing the true and the modified memory samples $\ket{\psi}_{\sigma}, \ket{\psi}_{\sigma*}$, with $\braket{{\psi}_{\sigma} | {\psi}_{\sigma*}}^2 \leq \epsilon$  we could apply $U$ to pry these two vectors apart, that is, $U\ket{\psi}_{\sigma*} =  \ket{\psi'}_{\sigma*},U \ket{\psi}_{\sigma} =  \ket{\psi'}_{\sigma} $ would have $\braket{{\psi'}_{\sigma} | {\psi'}_{\sigma*}}^2 \geq \epsilon$. Unfortunately, it is straightforward to show that this is impossible. 

    \begin{lemma}
    It is impossible to design a circuit $U$ where given two memory samples $\ket{\psi}_{\sigma*}, \ket{\psi}_{\sigma}$ such that $\braket{{\psi}_{\sigma} | {\psi}_{\sigma*}}^2 \leq \epsilon$, $U \ket{\psi}_{\sigma*} =  \ket{\psi'}_{\sigma*},U \ket{\psi}_{\sigma} =  \ket{\psi'}_{\sigma} $, $\braket{{\psi'}_{\sigma} | {\psi'}_{\sigma*}}^2 \geq \epsilon$
    \end{lemma}

\begin{proof}
    Given that $U$ is a quantum circuit, $U$ is a unitary matrix, which follows from the well known linearity of quantum mechanics. This implies $U^{\dagger} U = I$. $\braket{{\psi'}_{\sigma} | {\psi'}_{\sigma*}}^2 = \braket{{\psi}_{\sigma} U^{\dagger}| U {\psi}_{\sigma*}}^2 = \braket{{\psi}_{\sigma} | {\psi}_{\sigma*}}^2 \leq \epsilon$.

 \end{proof}

 Given the issues we have outlined above, we believe a fundamentally new approach is needed for attesting quantum memory; one that has no counterpart in the classical setting. We leave this problem for future work. 

\section{Leveraging quantum effects to counteract network attacks}
\label{sec-superdense}
We argue that the most important contribution of quantum theory to the field of remote memory attestation is the possibility of authenticating the prover without assuming trusted hardware. 
In particular, it seems possible by means of quantum effects to ensure that the prover has been active during the execution of the protocol; an authentication property known by \emph{aliveness}. 
    In \cite{Soteria}, the prover is forced to be alive in the challenge phase by relying on a PUF implementation in the prover. As we just showed, though, their approach suffers from various flaws. Hence, this section is dedicated to introducing a software-based remote attestation protocol that forces the prover to be alive during both the challenge phase and the response phase with minimal communication overhead. 

    \subsection{The protocol description}
    We present a protocol that conservatively extends the design in \cite{LaeuchliT24}. This allows us to illustrate the use of aliveness during the challenge phase without sacrificing security nor having to redo the security analysis given in \cite{LaeuchliT24}. 
    That is to say, we improve the protocol security by ensuring the prover is alive during the challenge phase without modifying existing security features of the protocol. The improvement lies specifically in the resistant of the protocol to proxy attacks, whereby the prover outsource the attestation task to a third-party device.  In terms of communication, our conservative extension requires the transmission of a few classical bits via superdense coding. We, however, retain roughly the same communication complexity by transmmitting only a few classical via superdense coding. Concretely, we transmit $2K$ classical bits where $K$ is a security parameter.   

    The protocol, depicted in Figure \ref{fig-protocol2}, consists of a setup and initial proximity-checking phase, multiple challenge-response phases, and a verification phase. We describe each of these phases next. 


	\noindent \emph{Setup and initial proximity-checking phase.} 
	The protocol starts with the verifier sending two nonces, $n^1_0$ and 
	$m^1_0$, 
	to 
	the prover. The former will be used in the next phase to feed a pseudo-random function $\gen$, 
	which 
	will create random memory addresses, the latter to calculate a checksum 
	$\chk$ 
	of the 
	prover's memory. 
    An important feature of this phase is that the prover quickly teleports $m_0^1$ back to the verifier. This allows the verifier to check whether the prover is close. 
	
	\noindent \emph{Challenge-response phase.}
    The challenge-response phase is executed $\kappa$ times, where $K$ is a security parameter. At the $i$th loop cycle, 
    the verifier picks two random bits $c_ic'_i$, prepares the qubit $\ket{\psi_i}$ to transmit via superdense coding, stores the current time $t_i$, and transmits $\ket{\psi_i}$. 
    Upon reception, the prover retrieves $c_ic'_i$ and proceeds to attest its memory by calculating a checksum of $N$ randomly chosen memory blocks. 
    
    The calculation of the checksum is performed in $N$ steps, where $N$ is a security parameter. At the $j$th step, 
    the function $\gen(n^i_{j-1} || c_i c'_i)$ is used to generate a pair of nonces $(n^i_j, 
	a^i_j)$, $s_j^i$ is used to store the 
	memory 
	word at the address $a^i_j$, and $m^i_j$ is used to store the the checksum of $m^{i-1}_j, s_j^i$ and 
	$r_{i-1}$. These checksums are encoded into a series of angles $\theta$,via normalisation into $[0, \pi]$. The prover then teleports the 
    qubit  $\ket{\Psi'} = sin(\frac{\theta'}{2}) 
    \ket{0} + 
    cos(\frac{\theta'}{2}) \ket{1}$. The teleportation is carried out by 
    entangling $\ket{\Psi}$ with the qubit to be teleported, then measuring in 
    the Bell basis.

	
	\noindent \emph{Verification phase.} The verifier checks that, for every round $i \in \{1, \ldots, N\}$, the checksum $m^i_N$ is correct 
	and the round-trip-time $t'_{i} - t_i$ is below a time threshold $\timethreshold$. 
    To verify that the checksum is correct, the verifier
    rotates the result based on the rotation bits ($r$) provided by the 
    prover in its response and angle $\theta$. Let $\ket{\Psi}$ be the 
    resulting qubit. The 
    verifier checks that $\ket{\Psi} = \ket{\Psi'}$ by performing an inverse 
    rotation $\ket{\Psi}$ 
    with 
    angle $\theta$. After measuring the verifier should obtain
    $\ket{1}$. If the verifier measures $\ket{0}$, then the memory has been compromised.

	\begin{figure}
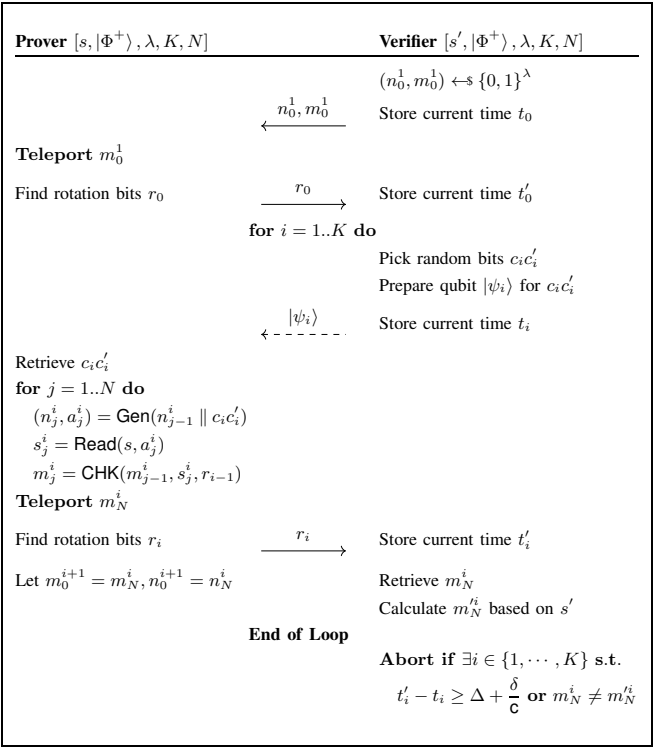

		\begin{center}
  \scalebox{0.75}{
			\fbox{
				\pseudocode[colspace=0cm]{%
					\< \< 
					\\
					\textbf{Prover $[s, \ket{\Phi^+}, \secparam, K, N]$} \< \< 
					\textbf{Verifier 
						$[s', \ket{\Phi^+}, \secparam, K, N]$}  
					\\[0.1\baselineskip][\hline]
					\<\< \\[-0.5\baselineskip]
					\< \< (n^1_0, m^1_0) \sample \bin^{\secparam} \\
					\<\< \\[-4ex]
					\< \sendmessageleft*[1.5cm]{n^1_0, m^1_0} \< \text{Store 
					current time $t_0$} \\
					\highlightkeyword{Teleport}	m^1_0 \<\<  
					\\
					\text{Find rotation bits } r_0 \< 
					\sendmessageright*[1.5cm]{r_0} \< 
					\text{Store 
						current time $t'_0$} \\
					\< \pcfor i = 1 .. K \pcdo \< \\
					\<\< \text{Pick random bits } c_ic'_{i} \\
					\<\< \text{Prepare qubit $\ket{\psi_{i}}$ for $c_ic'_{i}$} \\
					\< \sendmessageleft*[1.5cm, style ={dashed}]{\ket{\psi_{i}}} \< \text{Store 
					current time $t_i$} \\					
                    \text{Retrieve $c_ic'_{i}$} \<\< \\
					\pcfor j = 1 .. N \pcdo \< \< \\
                    \t (n^i_j, a^i_j) = \gen(n^i_{j-1} \mathbin\Vert c_ic'_{i}) \<\< \\
					\t s^i_j = \readf(s, a^i_j) \<\< \\
					\t m^i_j = \chk(m^i_{j-1}, s^i_j, r_{i-1}) 
					\< 
					\< \\
					\highlightkeyword{Teleport}	m^i_N \<\<  
					\\
					\text{Find rotation bits } r_{i} \< 
					\sendmessageright*[1.5cm]{r_{i}} 
					\< 
					\text{Store 
						current time $t'_i$} \\
					\text{Let } m^{i+1}_0 = m^{i}_N, n^{i+1}_0 = n^{i}_N \< \< 
					\text{Retrieve $m^i_N$} \\
					\< \< 
					\text{Calculate $m'^i_N$ based on $s'$} \\
					\< 
					\textbf{End of Loop} \< \\
					\<\< \highlightkeyword{Abort} \highlightkeyword{if} \exists i \in \{1, \cdots, K\} \ \highlightkeyword{ s.t.}
					\\
					\<\< \t t'_i - t_i \geq  
					\timethreshold + \frac{\maxdist}{\speed} \ \highlightkeyword{ or } m^i_N \neq m'^i_N \\
				} 
			}
   }
			\caption{A software-based memory attestation protocol that employs superdense coding and teleportation to ensure the prover is alive during both the challenge and the response phase. 
				\label{fig-protocol2}}	
		\end{center}
	\end{figure}
	
    It is worth noting that the round-trip-time $t'_0 - t_0$ is lower than the other round-trip-times $t'_i - t_i$ with $i > 1$, because the first round-trip-time consists of a prover reflecting back the challenge sent by the 
	verifier while the others include the time required to compute a partial memory checksum. Therefore, in practice, the time threshold for $t'_0 - t_0$ should be based purely on the speed of the communication channel, like in classical distance bounding. The time threshold for $t'_i - t_i$ with $i > 1$, however, will need to account for the computational time taken by attestation task.

\subsection{A note on the security of the protocol}

Because the introduced protocol is a conservative extension of the protocol in \cite{LaeuchliT24}, we claim it preserves the security bounds of the latter. In practice, however, the introduced protocol is more effective at counteracting proxy attacks. The reason being that it forces the prover to interact with the network attacker during both the challenge phase and the response phase, while the protocol in \cite{LaeuchliT24} \emph{only} forces the prover to participate in the response phase. The additional prover-attacker interaction increases the round-trip-time measured by the verifier, thereby forcing the attacker to be closer to the verifier.

\subsection{A note on the communication cost of our algorithm}

 While our method avoids the cost of quadratic communication between the Verifier and the Prover at the start of the algorithm, we still have a quadratic amount of communication between the Prover and Verifier during the verification phase.  However, we do manage a constant improvement over the swap test required by the algorithms of \cite{Soteria}. The reason being that the probability of success for this method is $|\braket{{\psi}_{\sigma}| {\psi}_{\sigma*}}|^2$, compared to $\frac{1}{2}-\frac{|\braket{{\psi}_{\sigma}| {\psi}_{\sigma*}}|^2}{2}$.  The improved efficiency is from the fact that the algorithm makes use of knowledge of what state the qubit is meant to be in, knowledge that is not exploited by the swap test. 

        \begin{figure}
  \caption{Comparison of the two different methods of detecting an attack }
  \centering
    \includegraphics[width=0.5\textwidth]{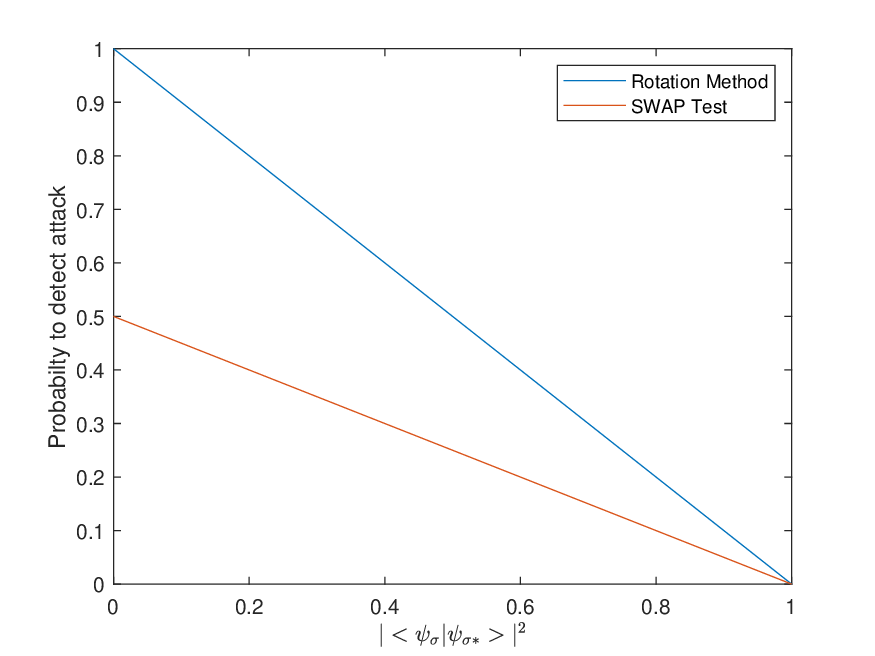}
\end{figure}

\section{Conclusions} 
 In this paper we have presented a method for remote memory attestation using quantum effects. This method combines the strengths of the proposals previously made in \cite{Soteria} and \cite{LaeuchliT24}. To achieve this we have analyzed the methods proposed in \cite{Soteria} extensively, and identified previously unknown and counter-intuitive limitations of these types of approaches, which we believe will be useful for anyone designing these types of protocols. We have used these observations to design a new protocol that makes it harder for the prover to outsource the attestation task to a far-away colluder, as the protocol forces the prover to be involve during both: the reception of the verifier's challenge and the sending of the response. 
Additionally, we remove the need for special purpose qPUFs, while retaining the strength they provide in this application area. This yields a device which is cheaper to produce and more easily implementable with near-term quantum computers which have a low circuit depth.

\bibliography{bib}
\end{document}